\documentclass[12pt,a4paper]{article}

\usepackage[margin=2.5cm]{geometry}
\usepackage{graphicx} 
\usepackage[T1]{fontenc}
\usepackage[utf8]{inputenc}
\usepackage{bm}
\usepackage{natbib}
\usepackage{float}
\usepackage{amsmath,amsfonts,amssymb,amsthm,mathabx}
\usepackage{soul}
\usepackage{array}
\usepackage{verbatim}
\usepackage{subfloat}
\usepackage{pdfpages}
\usepackage{url}
\usepackage{multirow}
\usepackage{setspace}
\usepackage{subfigure}
\usepackage{lscape}
\usepackage{booktabs}
\usepackage{xcolor}
\usepackage{xurl}
\usepackage[colorlinks=true, citecolor=blue, linkcolor=blue, urlcolor=blue]{hyperref}
\usepackage{ulem}

\newcommand{\bbeta}{\boldsymbol{\beta}}

\def \VV {{\mathbf V}}
\def \TT {{\mathbf T}}
\def \XX {{\mathbf X}}

\def \TT {{\mathbf T}}

\newcommand*{\bhat}{\skew{3}{\hat}{\bbeta}}

\newcolumntype{d}{D{.}{.}{2.5}}           

\newtheorem{proposition}{Proposition}

\providecommand{\keywords}[1]{\par\smallskip\noindent\textbf{Keywords:} #1}

\title{A Partial Fay--Herriot Model for Small Area Estimation: Estimating District-Level Consumption in Mozambique}

\author{%
Francesco Schirripa Spagnolo$^{1,\ast}$, Nicola Salvati$^{1}$, Enrico Fabrizi$^{2}$\\[1em]
\normalsize $^{1}$Department of Economics \& Management, University of Pisa,\\
\normalsize Via Cosimo Ridolfi 10, 56124 Pisa, Italy\\[0.3em]
\normalsize $^{2}$Department of Economic and Social Sciences, Catholic University of the Sacred Heart,\\
\normalsize Via Emilia Parmense 84, 29122 Piacenza, Italy\\[0.3em]
\normalsize $^{\ast}$Corresponding author: \href{mailto:francesco.schirripa@unipi.it}{francesco.schirripa@unipi.it}
}
\date{}

\begin{document}

\maketitle

\begin{abstract}
This paper proposes a new small area estimation approach that integrates Partial Least Squares within the Fay--Herriot model to address the challenges posed by high-dimensional and highly correlated auxiliary variables sets. The resulting Partial Fay--Herriot (PFH) estimator constructs supervised components that maximize their association with the target variable, enhancing model stability and predictive efficiency. Monte Carlo simulations demonstrate that PFH estimator achieves lower mean squared error than the standard Fay--Herriot estimator and outperforms principal components–based alternatives while relying on fewer latent dimensions. The methodology is applied to the estimation of district-level per capita consumption in Mozambique, where the survey data source is complemented by large set of correlated census variables. The resulting estimates highlight pronounced geographic heterogeneity and uncover spatial clusters of deprivation. Overall, the findings show that the proposed supervised dimension-reduction approach represents an effective and easily interpretable tool for producing reliable indicators in high-dimensional contexts.
\end{abstract}


\section{Introduction} \label{sec:intro}
Achieving the United Nations Sustainable Development Goals (SDGs) hinges on reducing poverty and inequalities including regional disparities. As a result, conducting thorough poverty assessments at a spatially disaggregated level helps pinpoint areas where implementing some supporting policies is essential. However, the existing surveys on income and living conditions are usually planned to provide reliable results only for large areas. Direct estimators, relying solely on domain-specific sample data, can lead to unreliable estimates for detailed disaggregations of the population. Expanding sample sizes of surveys is costly and time-consuming. Therefore, to address this challenge, small area estimation (SAE) methods have been developed. These methods enable obtaining estimates in specific small areas (or domains) with sufficient precision by combining survey data and other information sources, such as administrative, census, or other types of geographical data.

The phrase small area usually refers to any sub-population of the population of interest, such as geographical areas or socio-demographic groups \citep{rao2015}. SAE methodologies can be categorized into two main types: unit-level and area-level models. Unit-level models relate the values of a variable of interest to auxiliary information at the individual (observation) level; while area-level models are based on first computing direct estimators at the area level and then using information aggregated at the same level.
In recent years, these models have become more and more common. They do not require the assumption that the auxiliary variables are measured both in the survey and the auxiliary source and that their measurement is consistent \citep{tarozzi2009using}. Moreover, area level models are more apt to leverage diverse information sources available at the areal level or available only at the area level because of restrictions to their disclosure at the unit level. Eventually they enjoy in most cases a property of design-consistency that is appreciated by users and offers a protection against model failure \citep{de2024extended}. Both kinds of models usually rely on linear mixed models (LMMs) where the area-specific random effects capture the unobserved heterogeneity in the response variable across different areas. The most popular area-level model is the Fay--Herriot (FH), proposed in \cite{FH1979}. 

However, auxiliary data sources can include a lot of variables available for each area of interest and in some cases, their number ($p$) can exceed the number of areas $(m)$ itself, that is $p>m$. Quite likely, these auxiliary variables are largely redundant, i.e., they could be grouped in blocks of highly correlated variables. 

This is precisely the case in our study of Mozambique, a country where poverty remains widespread and regional disparities are particularly pronounced. Indeed, large segments of the population continue to face poor living conditions, with significant inequalities across districts \citep{egger2023, belchior2025}. Reliable poverty estimates at a fine geographical scale are therefore crucial for the government and international organizations to effectively design, monitor, and target poverty-reduction strategies.

In this setting, we focus on estimating average per capita consumption at the district level by combining survey and census data. The availability of a rich set of housing, demographic, and service-related covariates provides valuable predictive potential, but also raises serious issues of redundancy and multicollinearity.

To address these challenges, the researcher needs either to: \textit{i}) select a set of important variables among the large number of predictors, (which is unavoidable when $p>m$) \textit{ii}) apply regularization techniques (such as Ridge, LASSO, or Elastic Net), which allow simultaneous variable selection and shrinkage, or \textit{iii}) handle multicollinearity, as it makes standard regression approaches very inefficient. The use of dimension reduction methods represents a possible solution in this direction. Dimension reduction methods aim at reducing the numerous, redundant, original predictors by creating a small number of new ones (components) obtained as combinations of the former. Principal components analysis (PCA) represents, in this context, a traditional and popular approach. However, when PCA is applied, the new components are identified in an \textit{unsupervised} manner, as the outcome variable plays no role in determining the principal component directions. A \textit{supervised} alternative approach is given by Partial Least Squares (PLS) \citep{wold1966,helland1988} which uses the response to construct the directions of the new dimensions. As far as we know, no PLS-based methods have been developed and applied in the SAE context. 
An additional strength of our approach lies in the fact that it not only operates in a fully automatic manner but also has the potential to generate interpretable dimensions, which can further enhance the usefulness of the results.

To the best of our knowledge, no existing SAE methods integrate PLS in a supervised manner within the Fay--Herriot model. Our contribution is methodological, computational, and empirical: we construct a supervised dimension-reduction strategy tailored to area-level mixed models, derive its design-consistency, and demonstrate its advantages through extensive MC simulations and an application to Mozambique. The proposed approach provides a novel and effective tool to deliver reliable, disaggregated poverty indicators in Mozambique and can support the achievement of the SDGs in one of the world’s most vulnerable contexts. 

The remainder of this article is structured as follows. Section \ref{sec:data} describes the Mozambique consumption survey data used in our empirical application. Section \ref{sec:method} introduces the notation and provides a brief overview of FH model and the standard PLS approach, before presenting our proposed method for integrating PLS within the SAE framework. Section \ref{sec:simul} is devoted to assessing the performance of the proposed method through Monte Carlo (MC) simulations. In Section \ref{sec:appl}, we illustrate its practical relevance by estimating the average per capita consumption in the districts of Mozambique. Finally, Section \ref{sec:finalremarks} concludes the paper and outlines directions for future research.

\section{Data} \label{sec:data}

This study relies on the Inquérito sobre Orçamento Familiar (IOF, National Family Budget Survey) 2019/20, implemented by the Instituto Nacional de Estatística (INE) of Mozambique and made available through the World Bank. The IOF 2019/20 is a nationally representative household survey designed to measure household income, expenditure, consumption, and related socio-economic characteristics. It plays a central role in monitoring poverty and welfare dynamics in Mozambique, constitutes the basis for official poverty estimates, and is widely used in both research and policy analysis \citep{IOF}. 

Administratively, Mozambique is divided into 11 provinces, which represent the highest level of sub-national government. These provinces are further subdivided into districts. Historically, the country was composed of 154 districts, but following recent administrative reforms and the creation of new units, the 2017 Population and Housing Census reported a total of 161 districts \citep{INE2021}. Many household surveys, however, continue to adopt the earlier configuration of 154 districts, often excluding several in Cabo Delgado because of conflict-related inaccessibility \citep{ACAPS2023}.

Districts constitute the key level for planning and implementing public policies in Mozambique, as well as for monitoring poverty, consumption, and development outcomes. They function as development poles with dedicated funding, district development plans, and participatory monitoring and evaluation systems in place \citep{Massuanganhe2005}. In particular, our aim is to estimate the average per capita consumption expenditure at the district level.

Nevertheless, the IOF 2019/20 is designed to produce reliable estimates only at the national, urban/rural, and provincial levels. Districts are not planned domains of the survey, and direct estimation at this level is not feasible due to limited sample sizes. 
The final dataset comprises 13,302 individuals clustered within 154 districts. District-level sample sizes vary considerably, from a minimum of 1 to a maximum of 666 observations, with a median of 49.
As a preliminary step towards applying an SAE approach, we computed direct estimates of average per capita consumption expenditure at the district level. The standard errors of the direct estimator were obtained via analytical approximation using the function \texttt{direct} from the \texttt{sae} package in \texttt{R} \citep{molina-marhuenda:2015}.

The coefficients of variation (CVs) of the direct estimator are reported in Table \ref{tab.dir}. The classification of estimate reliability follows Statistics Canada guidelines: estimates with CV $< 16.6\%$ are considered reliable for general use; those with $16.6\% < \text{CV} \leq 33.3\%$ should be accompanied by cautionary notes; and those with CV $> 33.3\%$ are deemed unreliable \citep{canada}. Notably, in 19 out of 154 districts the CV exceeds 33.3\%, while in nearly half of the districts the CV lies between 16.6\% and 33.3\%. These results clearly highlight the necessity of adopting a SAE approach to reduce variability and improve the precision of district-level estimates.

\begin{table}[H]\small
	\caption{\label{tab.dir} Number of areas by CVs of the direct estimates.} \centering
	\begin{tabular}{lrrr}
		\toprule
		&\multicolumn{1}{c}{$CV < 16.6\%$} & \multicolumn{1}{c}{$16.6\% < CV < 33.3\%$} & \multicolumn{1}{c}{$CV > 33.3\%$}\\	
  \midrule
  Direct Est. &62&73&19\\
\bottomrule
\end{tabular}
\end{table}

As auxiliary information, we have access to data from the 2017 General Population and Housing Census. This dataset comprises 69 variables describing household composition, individual socio-demographic characteristics, and housing equipment, which constitute a valuable source of potential auxiliary variables for enhancing the SAE model. Consequently, an appropriate variable selection or dimension-reduction strategy is required to address the substantial redundancy among covariates. Indeed, the auxiliary variables exhibit block-wise correlation structures, as illustrated in Figure \ref{fig.corrplot}.

\begin{figure} [h!]
	\centering
	{\includegraphics[width=0.8\textwidth]{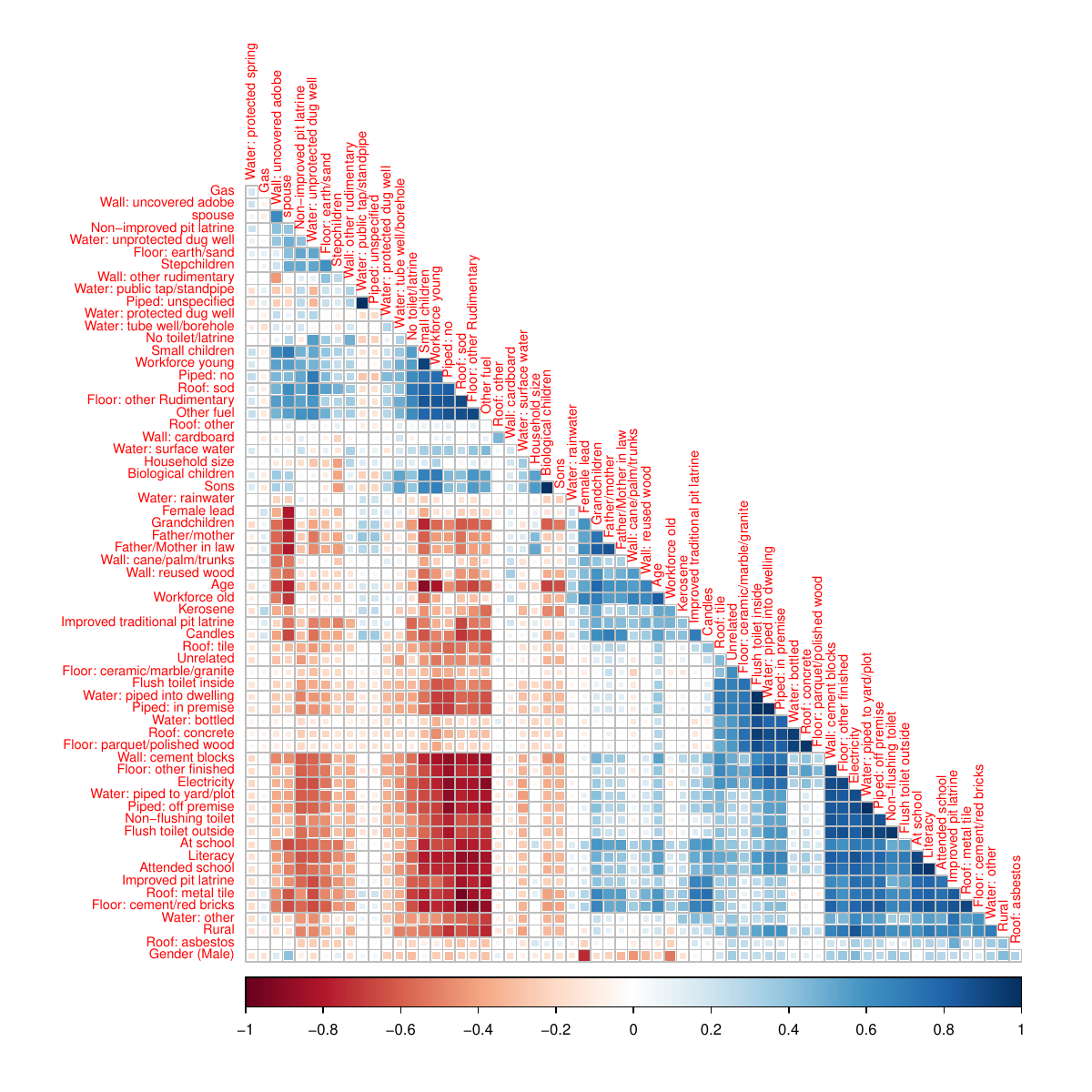}}	
	\caption{Correlation plot of the auxiliary variables from the Mozambique Census data.}
	\label{fig.corrplot}
\end{figure}

\section{Methodology} \label{sec:method}
In the following, we present the statistical methodology in which we combine a standard FH area-level model with the PLS procedure for dimension reduction. 

A finite population of size $N$ is assumed to be partitioned into $i = 1, 2, \ldots, m$ non-overlapping domains (i.e., small areas). The population size in each domain is equal to $N_i$. A sample of size $n_i$ is randomly drawn from the population values of the small area $i$ so that $n=\sum_{i=1}^{m}n_i$. 


We assume that, for each area of interest $i$ ($i=1,2,\ldots,m$), a direct survey estimator $\bar{Y}_{i}$ of the mean is available. Such estimators are typically design-based and rely solely on the sample information collected within the area.

The FH implies that a vector of $p$ auxiliary variables, $\XX_i=(X_{i1}, X_{i1},...X_{ip})^T$, and the parameters of interest, $\bar{Y}_i$, are linked by a certain relationship:
	
		\begin{equation} \label{areaRE}
		\bar{Y}_{i}=\XX_{i}^{T} \bbeta + u_i, \quad i=1,...m,
		\end{equation}
		
where $\bbeta$ is the vector of the regression coefficients
and $u_i$'s are area-specific random effects assumed to be independent and identically distributed with mean 0 and variance $\sigma_u^2$. 

Moreover, we assume that the following model holds for design unbiased direct estimator $\hat{Y}_i$:

		\begin{equation} \label{samplingmodel}
		\hat{\bar{Y}}_i=\bar{Y}_{i} + e_i, 
		\end{equation}
where $e_i$'s are the sampling errors in the $i^{th}$ area with $E(e_i|\bar{Y}_i) = 0$ and
$V(e_i|\bar{Y}_i) = \sigma_{e_i}^2$. 

Combining Equation \ref{areaRE} and Equation \ref{samplingmodel} we obtain the FH linear mixed model:

		\begin{equation} \label{FH}
		\hat{\bar{Y}}_i=\XX_{i}^{T} \bbeta + u_i + e_i. 
		\end{equation}

The Best Linear Unbiased Predictor of $\bar{Y}_i$ under the FH model is given by:

\begin{equation}
\hat{\bar{Y}}_i^{FH}=\gamma_i \hat{\bar{Y}}_{i} + (1-\gamma_i)  (\XX_{i}^{T} \bhat + u_i),
\end{equation}
where $\bhat$ is the weighted least square estimator of $\bbeta$ and $\gamma_i=\sigma_u^2/( \sigma_u^2 +\sigma_{e_i}^2)$, called the \textit{shrinkage factor}, measures the model variance relative to the total variance. Replacing $\sigma_u^2$ with a consistent estimator $\hat{\sigma}_u^2$ we obtain the EBLUP of $Y_i$.

However, in the context of high-dimensional auxiliary data and the presence of highly correlated variables, this model could be inefficient or impossible to estimate. As anticipated, a possible solution could be to use a dimension reduction method as the PLS.



The PLS methodology assumes the existence of a latent structure
\[ \mathbf{X} = \mathbf{T}\mathbf{P}^\top + \mathbf{G}, 
\qquad 
\mathbf{y} = \mathbf{T}\mathbf{b} + \mathbf{f},\]
where $\mathbf{X}\in\mathbb{R}^{l\times p}$ is the matrix of predictors, 
$\mathbf{y}\in\mathbb{R}^{l\times 1}$ is the univariate response,  $\mathbf{T}\in\mathbb{R}^{l\times H}$ is the score matrix formed by $H$ linear combinations of the original variables, 
$\mathbf{P}\in\mathbb{R}^{p\times H}$ is the loading matrix for $\mathbf{X}$,  $\mathbf{b}\in\mathbb{R}^{H\times 1}$ is the vector of regression coefficients,  and $\mathbf{G},\mathbf{f}$ are the corresponding error terms.  The decomposition is constructed to maximize the covariance between each latent component and the response.


We assume that both the predictors $\mathbf{X}$ and the response $\mathbf{y}$ are mean-centred (and possibly standardized). Let $\mathbf{X}_0=\mathbf{X}$ and $\mathbf{y}_0=\mathbf{y}$. 
For $h=1,\ldots,H$, the standard PLS algorithm proceeds iteratively as follows:

\begin{enumerate}
    \item Compute the weight vector
    \[
    \mathbf{w}_h = \frac{\mathbf{X}_{h-1}^\top \mathbf{y}_{h-1}}
    {\left\| \mathbf{X}_{h-1}^\top \mathbf{y}_{h-1} \right\|}, 
    \qquad \mathbf{w}_h \in \mathbb{R}^{p\times 1}.
    \]
    This standardized vector defines the direction in the predictor space along which the covariance between $\mathbf{X}$ and $\mathbf{y}$ is maximized. 
    It can therefore be interpreted as the linear combination of predictors that is most correlated with $\mathbf{y}$.
    \item Compute the PLS component 
    \[
    \mathbf{t}_h = \mathbf{X}_{h-1}\mathbf{w}_h, 
    \qquad \mathbf{t}_h \in \mathbb{R}^{l\times 1}.
    \]
    The latent component $\mathbf{t}_h$ is the projection of the predictor matrix onto the direction $\mathbf{w}_h$, forming a linear combination of covariates that best explains the variation in $\mathbf{y}$
    \item Calculate the $X$-loading vector as
    \[
    \mathbf{p}_h = \frac{\mathbf{X}_{h-1}^\top \mathbf{t}_h}{\mathbf{t}_h^\top \mathbf{t}_h}, 
    \qquad \mathbf{p}_h \in \mathbb{R}^{p\times 1}.
    \]
    The loading vector $\mathbf{p}_h$ measures how strongly each original variable contributes to the latent component $\mathbf{t}_h$.
    \item Calculate the $y$-loading for $\mathbf{y}$ as
    \[
    b_h = \frac{\mathbf{t}_h^\top \mathbf{y}_{h-1}}{\mathbf{t}_h^\top \mathbf{t}_h}, 
    \qquad b_h \in \mathbb{R}.
    \]
    The scalar $b_h$ expresses the regression coefficient of the current component $\mathbf{t}_h$ in the approximation of $\mathbf{y}$.
    \item Deflate the predictors and the response:
    \[
    \mathbf{X}_h = \mathbf{X}_{h-1} - \mathbf{t}_h \mathbf{p}_h^\top, 
    \qquad \mathbf{X}_h \in \mathbb{R}^{l\times p},
    \]
    \[
    \mathbf{y}_h = \mathbf{y}_{h-1} - b_h \mathbf{t}_h, 
    \qquad \mathbf{y}_h \in \mathbb{R}^{l\times 1}.
    \]
    This step removes from both $\mathbf{X}$ and $\mathbf{y}$ the information  already captured by the extracted component, ensuring that subsequent components explain new and uncorrelated sources of variation.
\end{enumerate}
After $H$ steps, collecting $\mathbf{W} = [\mathbf{w}_1,\ldots,\mathbf{w}_H]\in\mathbb{R}^{p\times H}$, 
$\mathbf{P} = [\mathbf{p}_1,\ldots,\mathbf{p}_H]\in\mathbb{R}^{p\times H}$ 
and $\mathbf{b}=(b_1,\ldots,b_H)^\top\in\mathbb{R}^{H\times 1}$, 
the estimated regression coefficients for the relation $\mathbf{y}\approx \mathbf{X}\hat{\boldsymbol\beta}$ are obtained as
\[
\hat{\boldsymbol\beta} = \mathbf{W}\,(\mathbf{P}^\top \mathbf{W})^{-1}\,\mathbf{b}, 
\qquad \hat{\boldsymbol\beta}\in\mathbb{R}^{p\times 1}.
\]






Our strategy is to retain the same algorithmic structure of the standard PLS procedure, while explicitly accounting for the hierarchical structure of the data. To this end, we replace the computation of the covariance between the outcome and the auxiliary variables with the adjusted $R^2$ measure for Fay--Herriot (FH) models proposed by \cite{lahiri2015}, and we compute the $y$-loadings in step 4 by fitting a FH model.

Let $\mathbf{X}_0=\mathbf{X}\in\mathbb{R}^{m\times p}$ denote the matrix of auxiliary variables and $\bar{Y}_0=\bar{Y}\in\mathbb{R}^{m\times 1}$ the vector of direct estimators. For $h=1,\ldots,H$ Partial Fay--Herriot (PFH) algorithm proceeds as follows:

Let $\mathbf{X}_0=\mathbf{X}\in\mathbb{R}^{m\times p}$ denote the matrix of auxiliary variables and $\bar{Y}_0=\bar{Y}\in\mathbb{R}^{m\times 1}$ the vector of direct estimators. For $h=1,\ldots,H$, the Partial Fay--Herriot (PFH) algorithm proceeds as follows:

\begin{enumerate}
\item Compute the weight vector $\mathbf{s}_h = (s_{h1}, \ldots, s_{hp})^\top \in \mathbb{R}^{p\times 1}$, 
	where each element $s_{hj}$ corresponds to the adjusted $R^2$ statistic from a univariate Fay--Herriot model
	fitted using $x_j$ as the only covariate:
	\[
	s_{hj} = 1 - \frac{h(\text{MSE}_j, \bar{D}_w)}{h(\text{MST}_j, \bar{D})},
	\qquad
	h(x,b) = \frac{2x}{1+\exp(2b/x)},
	\]
	with $\text{MSE}_j$ and $\text{MST}_j$ denoting, respectively, the mean squared error and mean squared total
	from the corresponding regression, $\bar{D}$ the average of the sampling variances, and $\bar{D}_w$
	their leverage-weighted mean as defined in \citet{lahiri2015}. The resulting vector $\mathbf{s}_h$ reflects the explanatory power of each auxiliary variable in the FH context. Finally, standardize to obtain
	\[
	\mathbf{w}_h = \frac{\mathbf{s}_h}{\|\mathbf{s}_h\|}, \qquad \mathbf{w}_h\in\mathbb{R}^{p\times 1}.
	\]
    \item Compute the PFH component
    \[
    \mathbf{t}_h = \mathbf{X}_{h-1}\mathbf{w}_h, \qquad \mathbf{t}_h\in\mathbb{R}^{m\times 1}.
    \]
    \item Compute the $X$-loading vector
    \[
    \mathbf{p}_h = \frac{\mathbf{X}_{h-1}^\top \mathbf{t}_h}{\mathbf{t}_h^\top \mathbf{t}_h}, 
    \qquad \mathbf{p}_h \in \mathbb{R}^{p\times 1}.
    \]
    \item Fit a FH model with $\bar{Y}_{h-1}$ as the response and $\mathbf{t}_h$ as covariate, obtaining the $Y$-loading 
    \[
    b_h \in \mathbb{R},
    \]
    and the corresponding vector of estimated random effects 
    \[
    \mathbf{u}_h\in\mathbb{R}^{m\times 1}.
    \]
    \item Deflate the predictors and the response as
    \[
    \mathbf{X}_h = \mathbf{X}_{h-1} - \mathbf{t}_h \mathbf{p}_h^\top, 
    \qquad \mathbf{X}_h\in\mathbb{R}^{m\times p},
    \]
    \[
    \bar{Y}_h = \bar{Y}_{h-1} - b_h \mathbf{t}_h - \mathbf{u}_h,
    \qquad \bar{Y}_h\in\mathbb{R}^{m\times 1}.
    \]
\end{enumerate}



At each iteration, the PFH algorithm identifies a latent direction that maximizes the explanatory capacity of the covariates under the FH model rather than the simple covariance with $\bar{Y}$ as in standard PLS. This allows the method to incorporate the hierarchical structure of the data and the sampling variances.

After extracting $H$ components, collected in $\mathbf{T}=[\mathbf{t}_1,\ldots,\mathbf{t}_H]\in\mathbb{R}^{m\times H}$, we fit a FH model with $\mathbf{T}$ as covariates to obtain the Partial Fay--Herriot (PFH) estimator

\begin{equation} \label{pfh_model}
   \hat{\bar{Y}}_i^{PFH} = \hat{\phi}_i \hat{\bar{Y}}_i + (1-\hat{\phi}_i)\bigl(\mathbf{T}_{i}^{\top}\hat{\boldsymbol\beta}^{PFH} + \hat{u}_i^{PFH}\bigr), 
\end{equation}
where $\hat{\boldsymbol\beta}^{PFH}\in\mathbb{R}^{H\times 1}$ and $\hat{u}_i^{PFH}$ are the regression coefficients and random effects obtained from the FH model with the PFH components as covariates, and $\hat{\phi}_i$ is the estimated shrinkage factor.



Design consistency is an important property for small area estimators as it guarantees that for large area-specific sample sizes, the estimator converges to the direct estimator, thereby offering a form of protection against model failure. We adopt the definition of design consistency due to    \cite{isaki1982survey}: an estimator sequence $(\widehat{\theta}_k)$ is said to be design‑consistent for $\theta$ if
\[
\hat\theta_k \xrightarrow{P_{p_k}} \theta_k 
\qquad \text{as } n_k \to \infty,
\]
that is, $
P_{p_k}\left( \, |\hat\theta_k - \theta| > \varepsilon \, \right) 
\to 0$, for every $\varepsilon > 0$.

For this definition to make sense in the finite population context, we have to posit the existence of a sequence of designs $p_k$ and finite populations $U_k$ of increasing size in a relationship endowed with regularity conditions including the strict positivity of first order inclusion probabilities, the uniform boundedness of second order inclusion probabilities, a \textit{stable} sampling rate, possibly going to 0 but not extremely fast. \citep[See][for a formal definition]{isaki1982survey}. The convergence is along the sequence of designs and population sizes. These technical conditions are met when the actual design is well behaved without extreme weights, points with extremely large leverage and bounded cluster sizes. 
\begin{proposition}
\label{prop:PFH_design}
Let $\hat{\bar{Y}}_i^{PFH}$ denote the PFH estimator defined in \eqref{pfh_model}. 
We assume that the sampling design satisfies the conditions in \cite{isaki1982survey}, the number of clusters grows with the sample size but cluster sizes remain uniformly bounded (i.e. there exist constants $N_{min}, N_{max}$, $0 < N_{min} < N_{max} < N_k$ such that $N_{min}\leq N_{ik} \leq N_{max}$ $\forall i$ in sub-population $U_k$). We then assume that as $n \to \infty$ also $n_i \to \infty$. As a consequence we have
\[
\hat{\bar{Y}}_i^{PFH} \xrightarrow{p} \bar{Y}_i \quad \text{as } n \to \infty,
\]
that is, the PFH estimator is design-consistent.
\end{proposition}
\begin{proof}
The PFH estimator is given by
\[
\hat{\bar{Y}}_i^{PFH} 
= \hat{\phi}_i \hat{\bar{Y}}_{i} 
+ (1-\hat{\phi}_i)\bigl(\TT_i^{\top}\hat{\beta}^{PFH} + \hat{u}_i^{PFH}\bigr),
\qquad 
\hat{\phi}_i = \frac{\hat{\sigma}_{u^{PFH}}^2}{\hat{\sigma}_{u^{PFH}}^2 + \sigma_{e_i}^2}.
\]
As the within-area sample size $n_i \to \infty$, the sampling variance $\sigma^2_{e_i} \to 0$, so that $\hat{\phi}_i \to 1$, while $\hat{\sigma}_{u^{PFH}}^2 \to \sigma^2_{u}$. 
Hence,
\[
\hat{\bar{Y}}_i^{PFH} \;\to\; \hat{\bar{Y}}_i.
\]
Since the direct estimator $\hat{\bar{Y}}_i$ is design-consistent for $\bar{Y}_i$, 
it follows that $\hat{\bar{Y}}_i^{PFH}$ is also design-consistent.
\end{proof}

\subsection{Parametric bootstrap MSE estimation}
\label{sec:bootstrap}

We propose a parametric bootstrap MSE estimator for the PFH predictor, following the general approach developed for EBLUP-type small area estimators by \citet{butar2003} and \citet{gonzalezmanteiga2008}. At each replicate, the bootstrap re-estimates the entire PFH procedure, both the component weights and the subsequent Fay--Herriot regression, so that the resulting variability reflects the uncertainty arising from both estimation stages.

For a given number of components $H$, the algorithm proceeds as follows.

\begin{enumerate}
\item Fit the PFH model with $H$ components on the observed data to obtain $\hat{\bbeta}_H$ and $\hat\sigma^2_{u,H}$, taken as the population values for the bootstrap, together with the components $\TT_H$ extracted from the original auxiliary variables $\XX$.
\item For $b = 1,\ldots,B$:
\begin{enumerate}
\item generate bootstrap area effects $u_i^{*} \sim N(0,\hat\sigma^2_{u,H})$ and set $\theta_i^{*} = [1,\TT_{H,i}]\hat{\bbeta}_H + u_i^{*}$, where $\TT_{H,i}$ denotes the $i$-th row of $\TT_H$;
\item generate bootstrap direct estimates $y_i^{*} = \theta_i^{*} + e_i^{*}$, with $e_i^{*}\sim N(0,\sigma^2_{e_i})$ and $\sigma^2_{e_i}$ the original sampling variances;
\item using $(y^{*},\XX)$ only, re-run the full PFH procedure, re-extracting $H$ components and refitting the Fay--Herriot model, to obtain the bootstrap predictor $\hat\theta_i^{*}$.
\end{enumerate}
\item estimate the MSE as
\[
\widehat{\mathrm{MSE}}^{boot}_i = \frac{1}{B}\sum_{b=1}^{B}\left(\hat\theta_{i,b}^{*}-\theta_{i,b}^{*}\right)^2 .
\]
\end{enumerate}

Because $\theta_i^{*}$ is known by construction at each replicate, step 3 provides a direct, model-consistent measure of prediction error. 

We note that our construction is a single-level (naive) parametric bootstrap: $\hat{\bbeta}_H$ and $\hat\sigma^2_{u,H}$ are themselves treated as the true population values, without propagating their own estimation uncertainty; a double bootstrap correction \citep{hallmaiti2006} would in principle refine this remaining source of variability, at a substantially higher computational cost. A formal theoretical comparison of this estimator with alternative variance-estimation approaches, together with the double bootstrap extension, is left for future research. More broadly, existing analytical variance estimators proposed for PLS typically treat the estimated components as fixed once extracted \citep[see, e.g.,][]{martinez2018}, and by construction leave out the part of the variability associated with estimating those components; the bootstrap approach proposed here avoids this simplification by design.

\section{Simulation study} \label{sec:simul}

The simulation study is based on the framework of \citet{bazzoli2023}.

We consider $m = 80$ areas of interest and three scenarios with $p = 200$, $p = 500$, and $p = 800$ auxiliary variables, respectively. For each scenario, the auxiliary variables $\XX = (\XX_1, \XX_2, \XX_3, \XX_4)$ are partitioned into four blocks of sizes $0.40p$, $0.40p$, $0.10p$ and $0.10p$. Each row of block $\XX_k$, $k=1,\ldots,4$, is generated independently from a multivariate normal distribution with an autoregressive-type covariance structure, $\{\VV^k_{\XX}\}_{ij} = c_k \rho^{|i-j|}$, where $\rho = 0.9$ and $c_1 = 8$, $c_2 = 4$, $c_3 = 2$, $c_4 = 1$. Blocks 1 and 2 therefore contain the variables with the strongest marginal variance, while blocks 3 and 4 contain comparatively lower-variance variables.

The regression coefficients are set as $\bbeta = \{10, \{0\}_{0.40p}, \{0\}_{0.40p}, \{0.5\}_{0.10p}, \{0.5\}_{0.10p}\}$: all predictors in blocks 1 and 2 are noise variables with zero coefficients, while all predictors in blocks 3 and 4 are informative, with coefficient $0.5$. As a result, the blocks carrying no information about the response are precisely those with the highest variance, making this design deliberately unfavourable to unsupervised dimension-reduction methods such as PCA, which are driven only by the marginal variance of the predictors and not by their association with the response \citep{bazzoli2023}.

True area-level values are generated as $\theta_i = 10 + \XX_i^\top \bbeta + u_i$, $u_i \sim N(0, \sigma_u^2)$, where $\sigma_u^2 = \frac{\bar\sigma^2_{e}}{\mathit{TAU}}$ and $\bar\sigma^2_{e} = \frac{\sigma^2_{e_i}}{\mathit{SNR}^2}$,   with heteroskedastic sampling variances defined as
\(
\sigma^2_{e_i} = 2 \times \left(1 + \frac{i-1}{m-1}\right)
\), $\mathit{SNR} = 3$ and $\mathit{TAU} = 4$, following the same signal-to-noise and variance-ratio parametrization adopted by \citet{bazzoli2023}.

Unbiased direct estimators are then obtained as
\begin{equation}\label{dir.sim}
\hat{\theta}^{Dir}_i = \theta_i + e_i,
\end{equation}
where $e_i \sim N(0, \sigma^2_{e_i})$, where $\sigma^2_{e_i}$ has been rescaled to have mean $\bar\sigma^2_e$. 

We computed four different predictors: the Direct Estimator (DIR); the Partial Fay--Herriot model (PFH); the empirical best linear unbiased predictor (EBLUP) under the standard FH model, where covariates were selected through forward stepwise regression; and a variant based on principal components (PCA), in which the FH model is fitted using principal component factors as covariates. The principal component factors were obtained by fitting a principal component regression (PCR) using the \texttt{pcr} function from the \texttt{pls} package in \texttt{R}. To enable a fair comparison between the PFH and PCA approaches, we fixed the number of dimensions included in the models and evaluated their performance across different dimensional settings. Specifically, models were fitted with 1 to 5 latent components.

As performance indicators we compute:

\begin{align*}\label{RBIAS.RMSE}
	\text{Bias}_{i}=\frac{1}{S}\underset{s=1}{\overset{S}{\sum}}\left(\hat\theta_i^{(s)}- \theta_{i}^{(s)} \right),\quad
	\text{MSE}_{i}=\frac{1}{S}\underset{s=1}{\overset{S}{\sum}}\left(\hat\theta_i^{(s)}- \theta_{i}^{(s)} \right)^2.
    \quad 
    \mathrm{Eff}_i = \frac{\mathrm{MSE}_i}{\mathrm{MSE}_i^{Dir}} \times 100\%,
\end{align*}.

Table~\ref{tab:sim_results_m80} reports the median Bias, MSE, and relative efficiency over the areas, for the three scenarios, from 500 Monte Carlo replications.

\begin{table}[htbp]
\centering
\small
\caption{Median over the areas of Bias, MSE, and $\mathrm{Eff}$ (\%) relative to the direct estimator, from 500 Monte Carlo simulations under the three scenarios with $m=80$ and $p=200,500,800$.}
\begin{tabular}{lccccccccc}
\toprule
 & \multicolumn{3}{c}{$p = 200$} & \multicolumn{3}{c}{$p = 500$} & \multicolumn{3}{c}{$p = 800$} \\
\cmidrule(lr){2-4} \cmidrule(lr){5-7} \cmidrule(lr){8-10}
 & Bias & MSE & Eff.(\%) & Bias & MSE & Eff.(\%) & Bias & MSE & Eff.(\%) \\
\midrule
Dir     &  0.018 & 18.257 & 100.0 &  0.005 & 62.568 & 100.0 & -0.016 & 111.225 & 100.0 \\
EBLUP   &  0.026 & 11.381 &  62.3 &  0.020 & 53.201 &  85.0 &  0.001 & 106.233 &  95.5 \\
\midrule
\multicolumn{10}{l}{PCA} \\
\hspace{3mm}1 component  &  0.003 & 16.299 &  89.3 &  0.007 & 55.004 &  87.9 & -0.051 & 98.443 &  88.5 \\
\hspace{3mm}2 components &  0.008 & 16.153 &  88.5 &  0.012 & 54.826 &  87.6 & -0.074 & 97.766 &  87.9 \\
\hspace{3mm}3 components &  0.022 & 15.925 &  87.2 &  0.022 & 53.799 &  86.0 & -0.030 & 96.544 &  86.8 \\
\hspace{3mm}4 components &  0.012 & 15.425 &  84.5 &  0.001 & 52.996 &  84.7 & -0.075 & 96.357 &  86.6 \\
\hspace{3mm}5 components &  0.013 & 14.704 &  80.5 & -0.004 & 51.750 &  82.7 & -0.108 & 95.318 &  85.7 \\
\midrule
\multicolumn{10}{l}{PFH} \\
\hspace{3mm}1 component  &  0.005 &  8.402 &  46.0 & -0.003 & 37.386 &  59.8 & -0.078 & 74.378 &  66.9 \\
\hspace{3mm}2 components &  0.002 &  7.430 &  40.7 &  0.009 & 35.041 &  56.0 & -0.022 & 71.786 &  64.5 \\
\hspace{3mm}3 components &  0.000 &  7.520 &  41.2 &  0.027 & 35.977 &  57.5 &  0.007 & 72.581 &  65.3 \\
\hspace{3mm}4 components & -0.001 &  7.975 &  43.7 &  0.003 & 36.832 &  58.9 & -0.013 & 76.892 &  69.1 \\
\hspace{3mm}5 components &  0.007 &  8.527 &  46.7 &  0.041 & 38.200 &  61.1 & -0.041 & 79.970 &  71.9 \\
\bottomrule
\end{tabular}
\label{tab:sim_results_m80}
\end{table}

Bias is negligible for all estimators and across all scenarios. The direct estimator has, by construction, $\mathrm{Eff}=100\%$ in every scenario. EBLUP improves on the direct estimator, but its advantage shrinks markedly as $p$ increases: $\mathrm{Eff}$ rises from $62.3\%$ at $p=200$ (a $37.7\%$ MSE reduction) to $95.5\%$ at $p=800$ (only a $4.5\%$ reduction), as forward stepwise selection becomes increasingly unstable when the number of candidate covariates grows relative to the number of areas, typically retaining between 8 and 11 covariates on average, but occasionally more than 30.

PCA-based estimators achieve modest and fairly stable gains over the direct estimator, with $\mathrm{Eff}$ between about $80\%$ and $89\%$ (an $11\%$ to $19\%$ reduction), improving gradually with the number of components. PFH clearly outperforms both EBLUP and PCA in every scenario: with only two components, $\mathrm{Eff}$ is $40.7\%$ at $p=200$, $56.0\%$ at $p=500$, and $64.5\%$ at $p=800$, MSE reductions of $59.3\%$, $44.0\%$ and $35.5\%$ respectively, against an $\mathrm{Eff}$ no lower than $80.5\%$ (a $19.5\%$ reduction) for the best PCA specification (five components) in the corresponding scenario. The point MSE of PFH is minimized with just two latent components in all three scenarios.

The relative advantage of PFH over PCA, while sizeable throughout, narrows as $p$ grows: comparing the best specification of each method (PFH with two components against PCA with five), the MSE reduction of PFH over PCA is $49.5\%$ at $p=200$, $32.3\%$ at $p=500$, and $24.7\%$ at $p=800$. This attenuation is consistent with the increasingly unfavourable ratio between the number of areas and the number of candidate auxiliary variables ($p/m = 2.5, 6.25$ and $10$, respectively): estimating a small number of supervised directions from a growing pool of correlated candidates, using information from only $m=80$ areas, becomes intrinsically more difficult as $p/m$ increases. Even so, PFH retains a clear and non-negligible advantage over both competitors across the whole range of dimensions considered.

Table~\ref{tab:r2_results_m80} reports the $R^2$ of the response explained by the latent components, and Table~\ref{tab:varx_results_m80} the corresponding cumulative share of the variance of the auxiliary variables explained.

\begin{table}[h!]
\centering
\small
\caption{$m=80$. $R^2$ (\%) of the response explained by PCA and PFH components.}
\begin{tabular}{lcccccc}
\toprule
 & \multicolumn{2}{c}{$p = 200$} & \multicolumn{2}{c}{$p = 500$} & \multicolumn{2}{c}{$p = 800$} \\
\cmidrule(lr){2-3} \cmidrule(lr){4-5} \cmidrule(lr){6-7}
Dimensions & PCA & PFH & PCA & PFH & PCA & PFH \\
\midrule
1 &  3.72 & 84.61 &  5.84 & 81.23 &  6.15 & 79.06 \\
2 & 12.11 & 87.45 & 15.78 & 84.88 & 12.91 & 82.99 \\
3 & 22.81 & 88.81 & 25.81 & 87.42 & 21.81 & 86.82 \\
4 & 34.71 & 89.40 & 36.01 & 88.81 & 30.01 & 88.84 \\
5 & 50.23 & 90.23 & 43.43 & 89.58 & 36.23 & 89.56 \\
\bottomrule
\end{tabular}
\label{tab:r2_results_m80}
\end{table}

\begin{table}[h!]
\centering
\small
\caption{$m=80$. Cumulative share (\%) of the variance of $\XX$ explained by PCA and PFH components.}
\begin{tabular}{lcccccc}
\toprule
 & \multicolumn{2}{c}{$p = 200$} & \multicolumn{2}{c}{$p = 500$} & \multicolumn{2}{c}{$p = 800$} \\
\cmidrule(lr){2-3} \cmidrule(lr){4-5} \cmidrule(lr){6-7}
Dimensions & PCA & PFH & PCA & PFH & PCA & PFH \\
\midrule
1 & 11.43 &  6.81 &  6.53 &  4.42 &  5.09 &  3.43 \\
2 & 21.18 & 13.21 & 12.42 &  7.98 &  9.76 &  6.22 \\
3 & 29.67 & 19.51 & 17.82 & 11.53 & 14.12 &  9.10 \\
4 & 37.02 & 25.34 & 22.84 & 14.90 & 18.18 & 11.75 \\
5 & 43.55 & 30.67 & 27.44 & 18.10 & 22.01 & 14.26 \\
\bottomrule
\end{tabular}
\label{tab:varx_results_m80}
\end{table}

PFH achieves a substantially higher $R^2$ than PCA with far fewer components: a single PFH component explains between $79\%$ and $85\%$ of the variability of the response across the three scenarios, whereas PCA requires all five components to explain at most $50\%$. At the same time, PFH systematically explains a smaller share of the total variance of $\XX$ than PCA (Table~\ref{tab:varx_results_m80}). This is expected, and indeed by design: the noise blocks (1 and 2), although unrelated to the response, carry the largest marginal variance in $\XX$, and are precisely the directions an unsupervised method such as PCA is bound to prioritise \citep{bazzoli2023}. The combination of a much higher $R^2$ with a much lower share of $\XX$-variance explained is the clearest evidence of the information efficiency of the supervised approach.

Table~\ref{tab:rb_boot_m80} reports the relative bias (RB) of the bootstrap MSE estimator of Section~\ref{sec:bootstrap}, computed for PFH with $B=200$ replicates.

\begin{table}[h!]
\centering
\small
\caption{$m=80$. Relative bias (\%) of the \emph{bootstrap} MSE estimator ($B=200$).}
\begin{tabular}{lccc}
\toprule
 & $p = 200$ & $p = 500$ & $p = 800$ \\
\midrule
\multicolumn{4}{l}{\textit{PFH}} \\
\hspace{3mm}1 component  &   7.147 &  -0.632 & -4.416 \\
\hspace{3mm}2 components &  -0.658 &   0.826 & -0.564 \\
\hspace{3mm}3 components & -13.263 & -10.350 & -9.351 \\
\hspace{3mm}4 components & -16.195 & -12.684 & -13.767 \\
\hspace{3mm}5 components & -19.902 & -15.747 & -14.943 \\
\bottomrule
\end{tabular}
\label{tab:rb_boot_m80}
\end{table}

Table~\ref{tab:rb_boot_m80} reports the relative bias of the bootstrap MSE estimator for PFH. At two components, the practically relevant choice, since it minimizes the point MSE of PFH in every scenario (Table~\ref{tab:sim_results_m80}), the bootstrap RB is close to zero across all values of $p$ ($-0.7\%$, $+0.8\%$, and $-0.6\%$ for $p=200$, $500$, and $800$ respectively), indicating that the estimator is essentially unbiased where it matters most for this application. A negative bias emerges at higher $H$ and grows in magnitude with the number of components, though somewhat unevenly across $p$: at $p=200$ it goes from about $-13\%$ at three components to about $-20\%$ at five, while the corresponding ranges are smaller at $p=500$ ($-10\%$ to $-16\%$) and $p=800$ ($-9\%$ to $-15\%$). This pattern is consistent with the single-level nature of the bootstrap discussed in Section~\ref{sec:bootstrap}: $\hat{\bbeta}_H$ and $\hat\sigma^2_{u,H}$ are treated as known population values, so the additional uncertainty in their own estimation is not propagated; this omitted source of variability may accumulate as more components are added to the model. 

\section{Estimate average per capita consumption in the districts of Mozambique} \label{sec:appl}

The aim of this section is to apply the PFH approach to estimate average per capita consumption in the districts of Mozambique. Given the asymmetric distribution of per capita consumption expenditure, a logarithmic transformation is adopted (see Figure \ref{fig.hist}).

\begin{figure}[h!]
    \centering
    \subfigure[per capita consumption expenditure]
	{\includegraphics[width=0.45\textwidth]{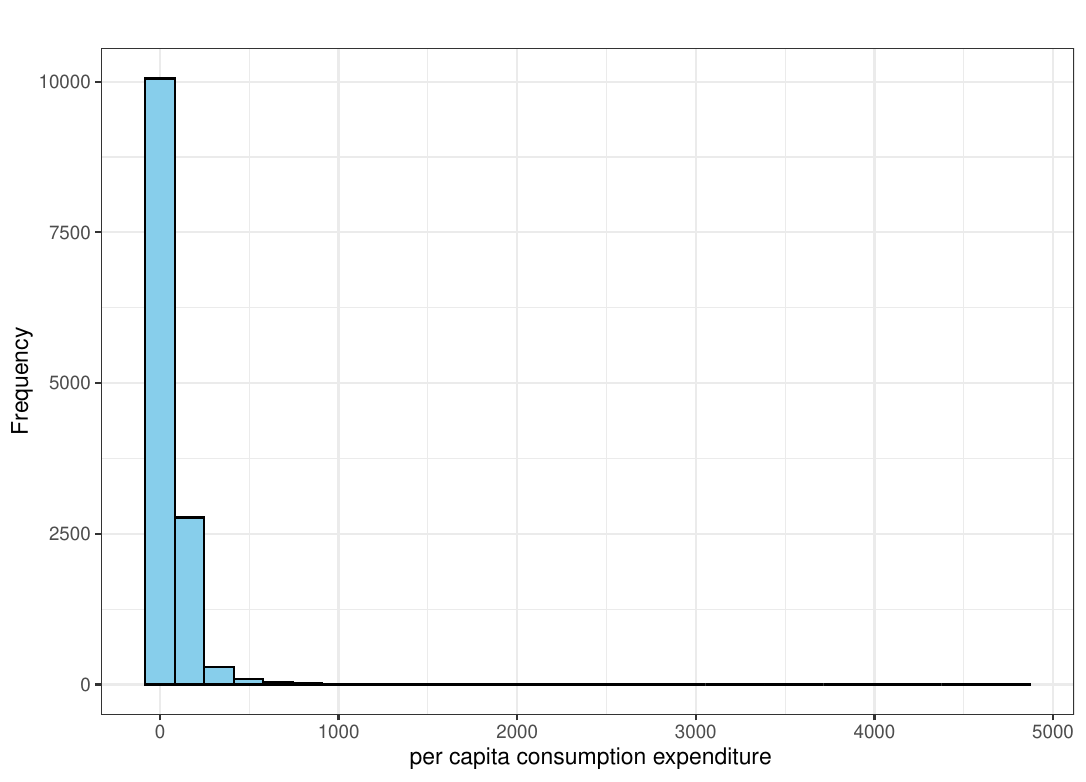}}	
    \hspace{1cm}
	\subfigure[log(per capita consumption expenditure)]
    {\includegraphics[width=0.45\textwidth]{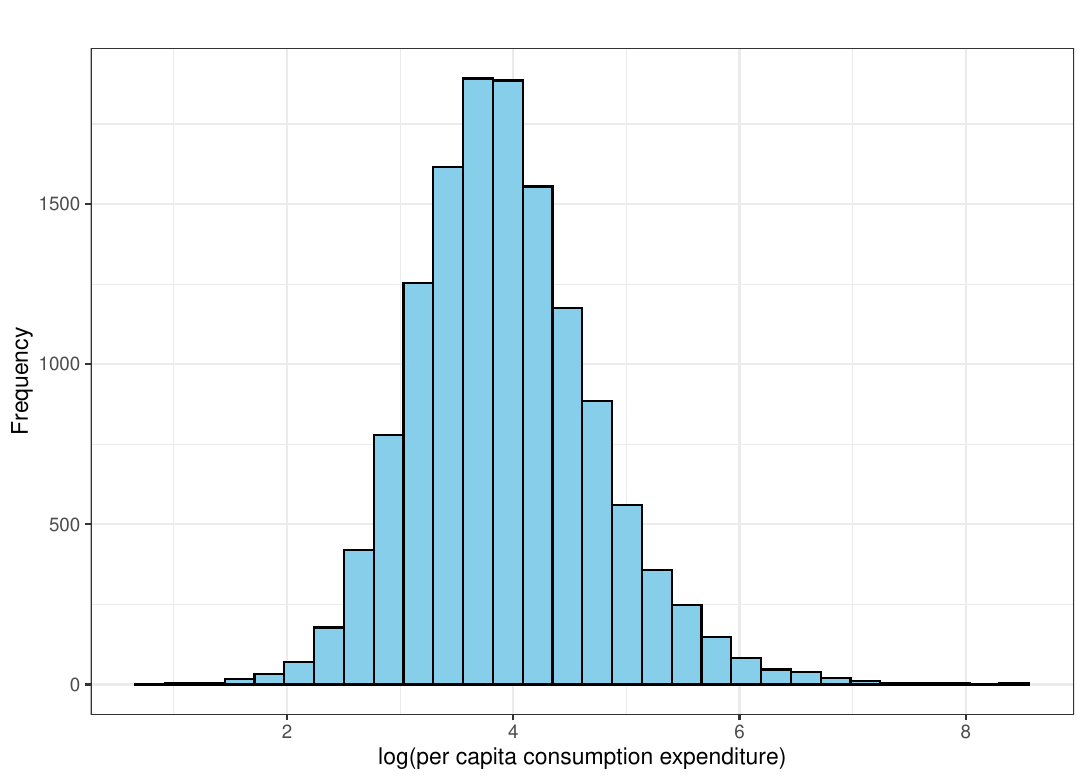}}
	\caption{Histogram of per capita consumption expenditure: (a) raw data and (b) logarithmic transformation. }
	\label{fig.hist}
\end{figure}

Following the evidence presented in the previous section, we employ the PFH methodology with two components to estimate district-level average per capita consumption. 
For comparison, we also consider the standard FH approach, where covariates are selected through a forward stepwise procedure, resulting in a set of ten variables.

The map in Figure \ref{fig.mapPFH} illustrates the spatial distribution of estimated average per capita consumption across the 154 districts of Mozambique.

\begin{figure}[h!]
    \centering
	{\includegraphics[width=0.95\textwidth]{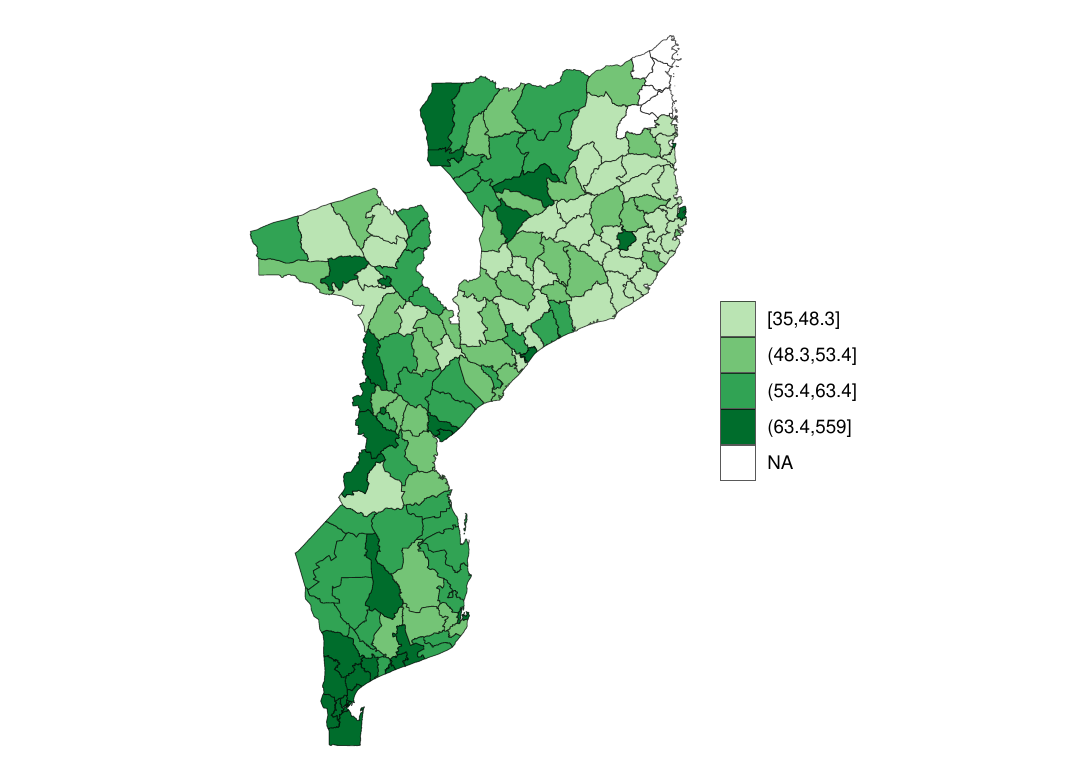}}	
	\caption{Estimated average per capita consumption across the 154 districts of Mozambique obtained using the PFH model.}
	\label{fig.mapPFH}
\end{figure}

The results reveal a clear geographic heterogeneity, with higher levels of per capita income concentrated in the southern districts, particularly around the capital region, and along some coastal areas. In contrast, much of the central and northern regions display lower average consumption, highlighting persistent regional disparities. Some inland districts in the north also register relatively higher values, which may be linked to specific economic activities or resource endowments. The map further indicates that poverty is not uniformly distributed, but instead exhibits clusters of low-income areas adjacent to pockets of relatively higher consumption. This spatial variability underscores the importance of adopting small area estimation techniques such as PFH, which can capture fine-grained differences often masked in aggregate national statistics. From a policy perspective, these results can help identify priority areas for targeted interventions, supporting more effective allocation of resources in poverty reduction programs.

Regarding the precision of the estimates, Table \ref{tab.CVpfh} reports the distribution of the coefficients of variation (CVs) for the PFH predictor. Out of the 154 districts, 128 (83.1\%) display CVs below 16.6\%, while 26 districts fall into the intermediate class between 16.6\% and 33.3\%. Notably, no district records a CV above 33.3\%. These results highlight a substantial improvement over the direct estimator (see Table \ref{tab.dir}).

\begin{table}[H]\small
	\caption{\label{tab.CVpfh} Number of areas by CVs of the direct estimates.} \centering
	\begin{tabular}{lrrr}
		\toprule
		&\multicolumn{1}{c}{$CV < 16.6\%$} & \multicolumn{1}{c}{$16.6\% < CV < 33.3\%$} & \multicolumn{1}{c}{$CV > 33.3\%$}\\	
  \midrule
  PFH Est. &128&26&0\\
\bottomrule
\end{tabular}
\end{table}

Model assumptions were also evaluated. The Q–Q plots in Figure \ref{fig.hist} suggest that the normality assumption is reasonable for both the unit-level errors and the district-level random effects.

\begin{figure}[h!]
    \centering
    \subfigure[]
	{\includegraphics[width=0.45\textwidth]{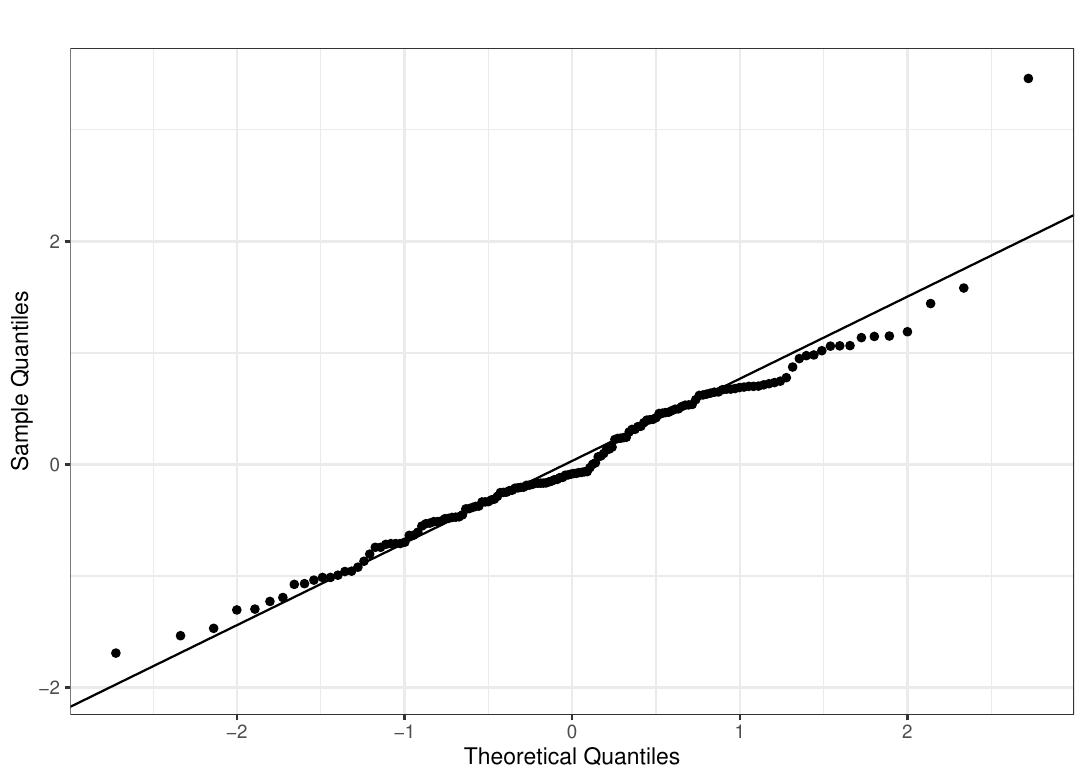}}	
    \hspace{1cm}
	\subfigure[]
    {\includegraphics[width=0.45\textwidth]{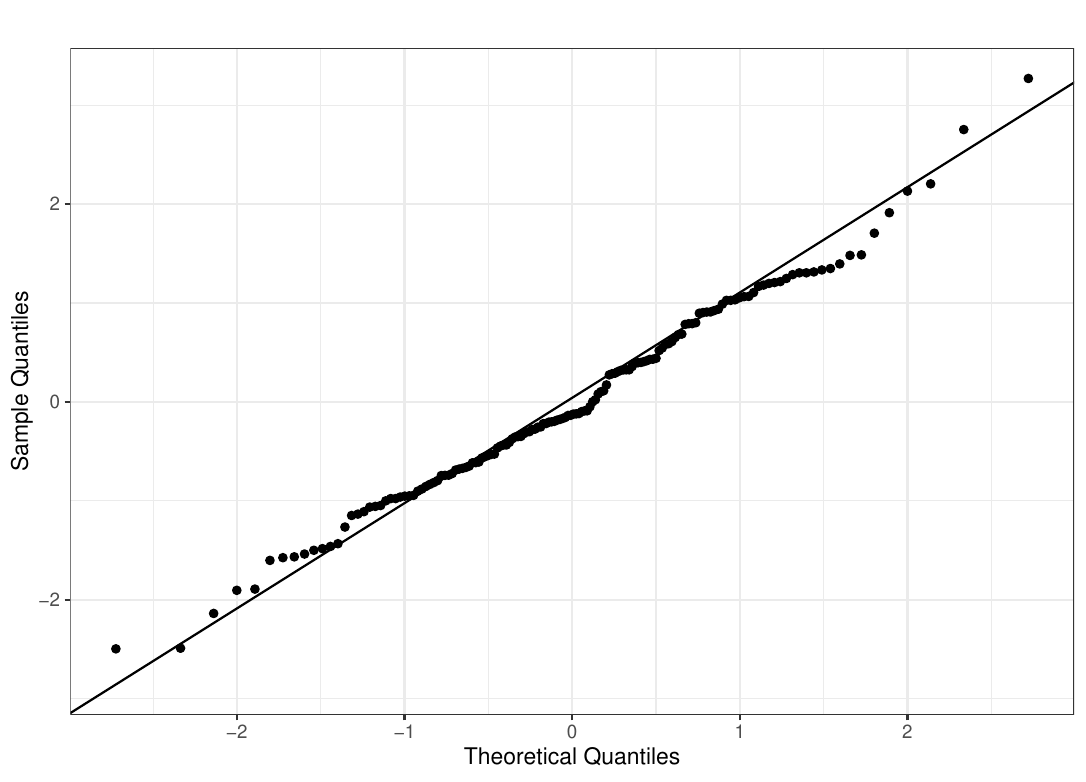}}
	\caption{Normal Q–Q plots for (a) the unit-level errors and (b) the area-level random effects obtained from the fitted PFH model.}
	\label{fig.hist}
\end{figure}

To further validate the model-based predictions, we rely on standard small area estimation diagnostics \citep{brown2001}. The correlation between direct and PFH estimates is 0.94, while the correlation between direct and EBLUP estimates under the FH model is equal to 0.93. Moreover, we computed the Wald-type goodness-of-fit statistic, defined as:

\begin{equation}\label{wald}
	W=\sum\limits_{d=1}^{m} \frac{(\hat{\theta}_i^{dir}-\hat{\theta}_{i}^{mod})^2}
	{{\hat{V}}(\hat{\theta}_i^{dir}) - {mse}(\hat{\theta}_{i}^{mod})},
\end{equation}

where $\hat{\theta}_i^{dir}$ corresponds to the direct estimate of the average consumption per capita in small area $i$; $\hat{V}(\hat{\theta}_i^{dir})$ is the estimated variance of the direct estimator; $\hat{\theta}_{i}^{mod}$ is the model-based estimate of the average income in small area $i$, with ${mse}(\hat{\theta}_{i}^{mod})$ denoting its corresponding estimated MSE.

Under the null hypothesis of equality between direct and model-based estimates, $W$ follows a chi-squared distribution with $m$ degrees of freedom. For PFH, we obtained $W=54.45$ with a $p$-value equal to 1, leading to the conclusion that direct and PFH-based estimates have the same expected value.



To gain further insights into the meaning of the PFH components, we analyzed the loadings reported in Figure \ref{fig.loadings}.

\begin{figure}[h!]
    \centering
	{\includegraphics[width=0.85\textwidth]{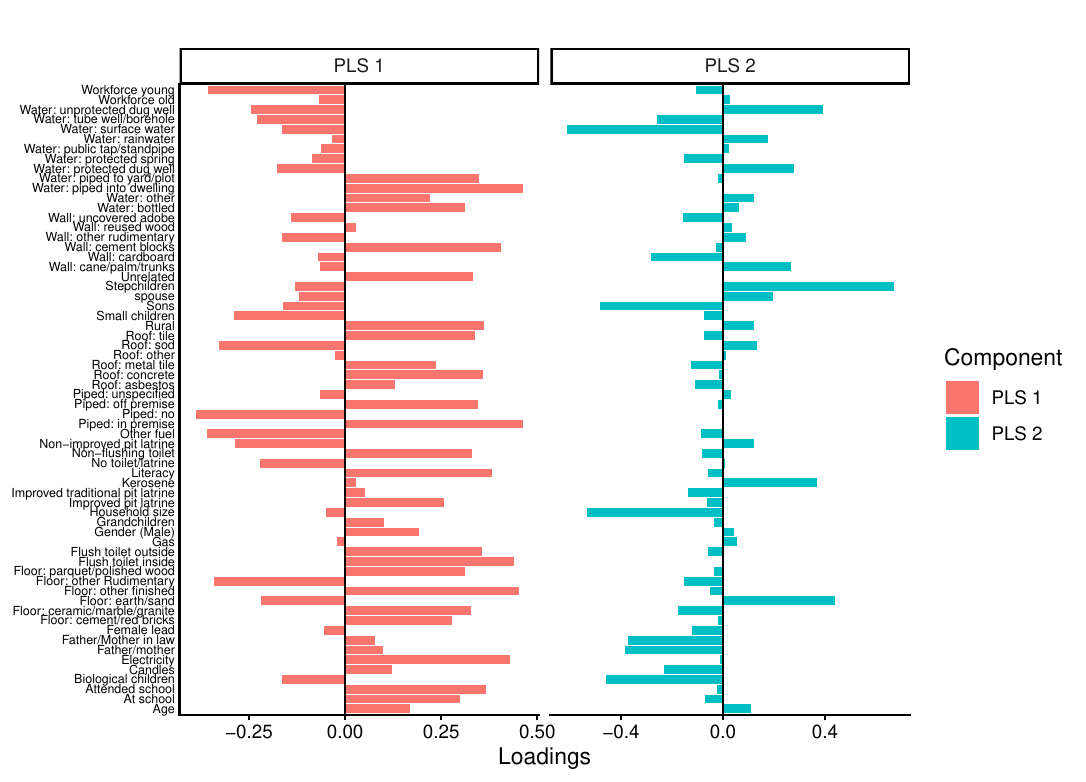}}	
	\caption{Loadings of the PFH components.}
	\label{fig.loadings}
\end{figure}

The first PLS component primarily captures household material well-being and access to essential infrastructure. High positive loadings are observed for piped water within the dwelling or on premises, finished floor materials, flush toilets inside the house, electricity, walls made of cement blocks, as well as literacy and school attendance. Conversely, negative loadings are found for the absence of piped water and rural location. This component aligns with the multidimensional poverty framework that emphasizes education and basic infrastructure services as key dimensions of welfare alongside monetary poverty \citep{d2024multidimensional}. Asset-based indicators such as housing quality and access to services like clean water and electricity have become standard proxies for household socioeconomic status in contexts where income and expenditure data are unreliable \citep{alkire2022revising}, particularly in sub-Saharan Africa where asset indices have been found to correspond well with monetary measures such as expenditure \citep{booysen2008using}. Namely, the DHS Wealth Index, constructed based on summarizing household assets and dwelling characteristics, has been widely adopted as a practical alternative to direct measurement of income or consumption expenditure, and has proven effective in representing long-term economic status \citep{rutstein2008}.

The second component reflects differences in household demographic structure and social vulnerability. It shows positive loadings for the presence of stepchildren, earth floors, unprotected dug wells, and the use of kerosene, whereas negative loadings are linked to larger household size, the presence of biological children, and co-resident parents or parents-in-law. This pattern contrasts socially fragile households -- smaller, with non-biological children and limited resources -- with more stable households characterized by larger size and conventional kinship structures. Research has established that family complexity, including the presence of stepchildren or half-siblings, is negatively associated with child well-being and academic outcomes, with stepchildren experiencing poorer outcomes than those living with only full siblings \citep{halpernmeekin2008}. Family structure instability and non-traditional household compositions have been shown to reproduce poverty through their association with reduced parental resources, increased maternal stress, and ultimately poorer child outcomes \citep{mclanahan2009}. Children living with cohabiting stepparents or relatives face substantially higher poverty rates compared to those living with married stepparents \citep{beegle2010orphanhood}, suggesting that household composition serves as an important marker of economic vulnerability.


Overall, the two PLS components capture two key dimensions of welfare in the Mozambican context: material living conditions and household demographic structure. This approach is consistent with multidimensional definitions of poverty that recognize multiple interlocked dimensions, with lack of access to basic infrastructure being particularly prominent in rural areas \citep{worldbank2018mozambique}. By incorporating both asset-based welfare indicators and family composition variables, the model recognizes that poverty involves being deprived on several fronts that do not necessarily correlate perfectly with monetary wealth \citep{alkire2022revising}, providing a more nuanced assessment of household welfare at the district level.

\section{Final remarks} \label{sec:finalremarks}

The United Nations’ 2030 Agenda for Sustainable Development, with its overarching commitment to leaving no one behind, relies on the availability of disaggregated, timely, and reliable statistics. This requirement is particularly urgent in countries such as Mozambique, where poverty levels vary widely across regions and where targeted local policies are essential to reduce territorial inequalities. SAE methods offer a powerful solution in this context, as they enable the production of reliable estimates for small geographical domains, supporting both policy design and impact evaluation.

However, modern auxiliary data sources, such as census, administrative, or geospatial datasets, often contain a very large number of variables for each area. These predictors are typically highly redundant and strongly correlated, making standard modelling approaches unstable and requiring researchers to adopt techniques capable of variable selection and effective handling of multicollinearity.

In this paper, we propose an automatic SAE approach that integrates PLS within the most common SAE area-level model. Unlike traditional dimensionality-reduction methods such as PCA, PLS constructs components that explicitly take into account their relationship with the target variable. This supervision not only enhances predictive performance but also contributes to the interpretability of the resulting components. The method operates in a fully automatic fashion and provides a transparent structure that helps practitioners understand which combinations of auxiliary variables are most relevant for prediction. Using the 69 covariates available from the General Population and Housing Census of Mozambique, our approach effectively estimates district-level per capita consumption expenditure.

The proposed methodology was validated through a MC simulation study. Results showed that our PLS-based SAE estimator performs well across the considered scenarios, improving the reliability of mean squared error estimation compared with the standard Fay--Herriot model, and achieving comparable explanatory power with fewer components than PCA. This highlights the advantage of supervised dimension reduction in small area applications, where parsimony and predictive accuracy are both essential.

The empirical application to Mozambique reveals substantial geographic heterogeneity in consumption levels, with lower average consumption observed in the central and northern regions. The SAE estimates further uncover clusters of low-income districts adjacent to pockets of relatively higher consumption, patterns that would be difficult to detect otherwise. Moreover, the analysis of the loadings associated with the PFH components provides insights into the most influential predictors, illustrating the interpretability benefits of the proposed approach.

The results obtained here suggest several promising avenues for future methodological development. One natural extension involves introducing robustness either in the dimension-reduction step or in the mixed-model estimation, to better accommodate outlying areas or irregular auxiliary information. Another important direction concerns the MSE estimation: although the parametric bootstrap approach performs well with one or two components, a negative bias emerges when the number of components increases. Double bootstrap extension could improve could more fully capture the uncertainty and its developments is left for future research. 
Spatial extensions also appear particularly relevant, as many socioeconomic indicators and auxiliary sources exhibit strong spatial dependence; embedding PFH within spatial Fay--Herriot models \citep{petrucci2006} would allow supervised dimension reduction that respects such structure.

Overall, the PFH approach offers a flexible, interpretable, and automatic solution to the challenges posed by high-dimensional auxiliary information in area-level SAE. Its strong performance in both simulated and real-data settings suggests that supervised dimension reduction can play a central role in modern SAE applications, supporting the production of reliable and policy-relevant small area indicators in increasingly complex data environments.


\section*{Acknowledgments}
The work of Nicola Salvati was carried out with the support of the project ``Quantification in the Context of Dataset Shift'' (QuaDaSh) (Grant P2022TB5JF, Italy). 
\\ \emph{Conflicts of interest}: The authors declare no conflicts of interest.


\section*{Data availability}
The data that support the findings of this study are available from World Bank group but restrictions apply to the availability of these data, which were used under license for the current study, and so are not publicly available. Data are however available from the authors upon reasonable request and with permission of World Bank group. 


\bibliographystyle{agsm} 
\bibliography{BibPFH}

@article{lahiri2015,
  title={Variable selection for linear mixed models with applications in small area estimation},
  author={Lahiri, Partha and Suntornchost, Jiraphan},
  journal={Sankhya B},
  volume={77},
  pages={312--320},
  year={2015},
  publisher={Springer}
}

@article{de2024extended,
  title={{Extended Beta models for poverty mapping. An application integrating survey and remote sensing data in Bangladesh}},
  author={De Nicol{\`o}, Silvia and Fabrizi, Enrico and Gardini, Aldo},
  journal={The Annals of Applied Statistics},
  volume={18},
  number={4},
  pages={3229--3252},
  year={2024},
  publisher={Institute of Mathematical Statistics}
}

@article{tarozzi2009using,
  title={Using census and survey data to estimate poverty and inequality for small areas},
  author={Tarozzi, Alessandro and Deaton, Angus},
  journal={The review of economics and statistics},
  volume={91},
  number={4},
  pages={773--792},
  year={2009},
  publisher={The MIT Press}
}

@misc{INE2021,
  author       = {INE},
  title        = {{IV Recenseamento Geral da População e Habitação 2017: Resultados definitivos}},
  year         = {2021},
  howpublished = {Instituto Nacional de Estatística, Maputo},
  note          = {\url{https://mozdata.ine.gov.mz/index.php/catalog/24} . [Online. Accessed 9 December 2025]}
}

@misc{ACAPS2023,
  author       = {{ACAPS}},
  title        = {{Mozambique: Impact of the five-year conflict in Cabo Delgado}},
  year         = {2023},
  howpublished = {ACAPS Thematic Report},
  month        = {July},
  note          = {\url{https://www.acaps.org/fileadmin/Data_Product/Main_media/20230707_ACAPS_Thematic_report_Mozambique_impact_of_the_five-year_conflict_in_Cabo_Delgado.pdf}. [Online. Accessed 9 December 2025]}
}

@misc{IOF,
  author       = {{INE}},
  title        = {{Inquérito sobre Orçamento Familiar -- IOF 2019/20: Relatório Final}},
  year         = {2022},
  howpublished = {Instituto Nacional de Estatística, Maputo},
  note          = {\url{https://university.open.ac.uk/technology/mozambique/sites/www.open.ac.uk.technology.mozambique/files/files/IOF%202019_20%20Final%2022_09_2021.pdf}. [Online. Accessed 9 December 2025]}
}

@techreport{Massuanganhe2005,
  author      = {Massuanganhe, Israel Jacob},
  title       = {{Decentralization and District Development: Framework for Decentralized Policies and Local Development Strategies}},
  institution = {{UNDP/UNCDF Mozambique}},
  year        = {2005},
  type        = {Working Paper},
  number={3}
}

@Article{molina-marhuenda:2015,
    author = {Isabel Molina and Yolanda Marhuenda},
    title = {{sae}: An {R} Package for Small Area Estimation},
    journal = {The R Journal},
    year = {2015},
    volume = {7},
    number = {1},
    pages = {81--98}
  }

@misc{canada,
  title={{Survey of Household Spending 2006: Data Quality Indicators}},
  author={{Statistics Canada}},
  year={2010}, 
  howpublished ={\url{https://www150.statcan.gc.ca/n1/en/pub/62f0026m/62f0026m2010003-eng.pdf?st=EiIPLz-D}},
note={[Online. Accessed 3 September 2025]}
}

@inproceedings{brown2001,
	author = {Brown, Gary and Chambers, Ray and Heady, Patrick and Heasman, Dick},
	booktitle = {Proceedings of {Statistics Canada Symposium} 2001 on {Achieving} data quality in a statistical agency: a methodological perspective},
	title = {Evaluation of small area estimation methods--an application to unemployment estimates from the {\scshape UK LFS}},
	year = {2001}}

@book{rao2015,
	address = {New York},
	author = {J. N. K. Rao and I. Molina},
	publisher = {2nd ed. Wiley},
	title = {Small Area Estimation},
	year = {2015}}

@article{FH1979,
  title={{Estimates of income for small places: An application of James-Stein procedures to census data}},
  author={Fay, Robert E. and Herriot, Roger A.},
  journal={Journal of the American Statistical Association},
  volume={74},
  number={366a},
  pages={269--277},
  year={1979},
  publisher={Taylor \& Francis}
}

@article{wold1966,
  title={Estimation of principal components and related models by iterative least squares},
  author={Wold, Herman},
  journal={Multivariate analysis},
  pages={391--420},
  year={1966},
  publisher={Academic Press}
}

@article{helland1988,
  title={On the structure of partial least squares regression},
  author={Helland, Inge S},
  journal={Communications in statistics-Simulation and Computation},
  volume={17},
  number={2},
  pages={581--607},
  year={1988},
  publisher={Taylor \& Francis}
}

@article{petrucci2006,
	author = {Petrucci, A. and Salvati, N.},
	journal = {Journal of Agricultural Biological and Environmental
Statistics},
	pages = {169--182},
	title = {{Small Area Estimation considering Spatial Correlation in Watershed Erosion Assessment Survey}},
	volume = {11},
	year = {2006}}

@article{martinez2018,
  title={A new estimator for the covariance of the {PLS} coefficients estimator with applications to chemical data},
  author={Mart{\'\i}nez, Jos{\'e} L and Leiva, V{\'\i}ctor and Saulo, Helton and Ruggeri, Fabrizio and Arteaga, Gean C},
  journal={Journal of Chemometrics},
  volume={32},
  number={12},
  pages={e3069},
  year={2018},
  publisher={Wiley Online Library}
}

@article{egger2023,
  title={{Evolution of multidimensional poverty in crisis-ridden Mozambique}},
  author={Egger, Eva-Maria and Salvucci, Vincenzo and Tarp, Finn},
  journal={Social Indicators Research},
  volume={166},
  number={3},
  pages={485--519},
  year={2023},
  publisher={Springer}
}

@techreport{belchior2025,
  title={{Spatial Dynamics and Convergence of Multidimensional Poverty in Mozambique}},
  author={Belchior, Manuel José  and Chagas, André Luis Squarize},
  year={2025},
  institution={N{\'u}cleo de Economia Regional e Urbana da Universidade de S{\~a}o Paulo (NEREUS)}
}

@article{halpernmeekin2008,
  author = {Halpern-Meekin, Sarah and Tach, Laura},
  title = {Heterogeneity in Two-Parent Families and Adolescent Well-Being},
  journal = {Journal of Marriage and Family},
  year = {2008},
  volume = {70},
  number = {2},
  pages = {435--451},
  doi = {10.1111/j.1741-3737.2008.00492.x}
}

@article{mclanahan2009,
  author = {McLanahan, Sara},
  title = {Fragile Families and the Reproduction of Poverty},
  journal = {The ANNALS of the American Academy of Political and Social Science},
  year = {2009},
  volume = {621},
  number = {1},
  pages = {111--131},
  doi = {10.1177/0002716208324862}
}

@techreport{rutstein2008,
  author = {Rutstein, Shea O.},
  title = {{The DHS Wealth Index: Approaches for Rural and Urban Areas}},
  institution = {Macro International Inc.},
  year = {2008},
  type = {DHS Working Papers},
  number = {60},
  address = {Calverton, Maryland}
}

@article{d2024multidimensional,
  title={Multidimensional poverty: an analysis of definitions, measurement tools, applications and their evolution over time through a systematic review of the literature up to 2019},
  author={D’Attoma, Ida and Matteucci, Mariagiulia},
  journal={Quality \& Quantity},
  volume={58},
  number={4},
  pages={3171--3213},
  year={2024},
  publisher={Springer}
}

@article{alkire2022revising,
  title={{Revising the global multidimensional poverty index: Empirical insights and robustness}},
  author={Alkire, Sabina and Kanagaratnam, Usha and Nogales, Ricardo and Suppa, Nicolai},
  journal={Review of Income and Wealth},
  volume={68},
  pages={S347--S384},
  year={2022},
  publisher={Wiley Online Library}
}

@article{beegle2010orphanhood,
  title={{Orphanhood and the living arrangements of children in sub-Saharan Africa}},
  author={Beegle, Kathleen and Filmer, Deon and Stokes, Andrew and Tiererova, Lucia},
  journal={World Development},
  volume={38},
  number={12},
  pages={1727--1746},
  year={2010},
  publisher={Elsevier}
}

@article{booysen2008using,
  title={{Using an asset index to assess trends in poverty in seven Sub-Saharan African countries}},
  author={Booysen, Frikkie and Van Der Berg, Servaas and Burger, Ronelle and Von Maltitz, Michael and Du Rand, Gideon},
  journal={World Development},
  volume={36},
  number={6},
  pages={1113--1130},
  year={2008},
  publisher={Elsevier}
}

@techreport{worldbank2018mozambique,
  author = {{World Bank}},
  title = {{Mozambique Poverty Assessment: Strong but Not Broadly Shared Growth}},
  institution = {World Bank},
  year = {2018},
  address = {Washington, DC},
  number = {130329-MZ},
note={\url{https://documents.worldbank.org/en/publication/documents-reports/documentdetail/248561541165040969}. [Online. Accessed 9 December 2025]}
}

@article{isaki1982survey,
  title={Survey design under the regression superpopulation model},
  author={Isaki, Cary T and Fuller, Wayne A},
  journal={Journal of the American Statistical Association},
  volume={77},
  number={377},
  pages={89--96},
  year={1982},
  publisher={Taylor \& Francis}
}

@article{bazzoli2023,
  title={Partial least square based approaches for high-dimensional linear mixed models: C. Bazzoli et al.},
  author={Bazzoli, Caroline and Lambert-Lacroix, Sophie and Martinez, Marie-Jos{\'e}},
  journal={Statistical Methods \& Applications},
  volume={32},
  number={3},
  pages={769--786},
  year={2023}
}

@article{butar2003,
  author  = {Butar, Ferry B. and Lahiri, Partha},
  title   = {On measures of uncertainty of empirical {B}ayes small-area estimators},
  journal = {Journal of Statistical Planning and Inference},
  year    = {2003},
  volume  = {112},
  number  = {1--2},
  pages   = {63--76}
}

@article{gonzalezmanteiga2008,
  author  = {Gonz{\'a}lez-Manteiga, W. and Lombard{\'\i}a, M. J. and Molina, I. and Morales, D. and Santamar{\'\i}a, L.},
  title   = {Bootstrap mean squared error of a small-area {EBLUP}},
  journal = {Journal of Statistical Computation and Simulation},
  year    = {2008},
  volume  = {78},
  number  = {5},
  pages   = {443--462}
}

@article{hallmaiti2006,
  author  = {Hall, Peter and Maiti, Tapabrata},
  title   = {On parametric bootstrap methods for small area prediction},
  journal = {Journal of the Royal Statistical Society: Series B},
  year    = {2006},
  volume  = {68},
  number  = {2},
  pages   = {221--238}
}

\end{document}